\documentclass[12pt, draftclsnofoot, onecolumn]{IEEEtran}
\usepackage{amsmath,amsfonts}
\usepackage{algorithmic}
\usepackage{algorithm}
\usepackage{array}
\usepackage[caption=false,font=normalsize,labelfont=sf,textfont=sf]{subfig}
\usepackage{textcomp}
\usepackage{stfloats}
\usepackage{url}
\usepackage{verbatim}

\usepackage{graphicx}
\usepackage{cite}
\usepackage{xcolor}

\newtheorem{lemma}{Lemma}
\newtheorem{remark}{Remark}
\newenvironment{proof}{ \textbf{Proof:} }{ \hfill $\Box$}

\newcommand{\equalityA}{\stackrel{(a)}{=}}

\def\bb0{{\mathbb{0}}}

\def\ba{{\mathbf{a}}}
\def\bb{{\mathbf{b}}}

\def\bff{{\mathbf{f}}}

\def\b0{{\mathbf{0}}}
\def\cG{\mathcal{G}}

\def\cO{\mathcal{O}}

\def\sf0{{\mathsf{0}}}
\def\rm0{{\mathrm{0}}}
\def\b0{{\pmb{0}}} 
\begin{document}

\title{A wideband generalization of the near-field region for extremely large phased-arrays}
%\title{Near-field behavior  of wideband systems  with extremely large antenna arrays}

\author{\IEEEauthorblockN{Nitish  Deshpande, \textit{Student Member, IEEE,} Miguel R. Castellanos, \textit{Member, IEEE,} Saeed R. Khosravirad, \textit{Member, IEEE,} Jinfeng Du,   \textit{Member, IEEE,}   Harish Viswanathan, \textit{Fellow, IEEE,} and Robert W. Heath Jr., \textit{Fellow, IEEE}}\\ 
%	\IEEEauthorblockA{{Department of Electrical and Computer Engineering, North Carolina State University, Raleigh, NC} \\
%		Email: \{nvdeshpa, mrcastel, rwheathjr\}@ncsu.edu}
	\thanks{ Nitish  Deshpande,  Miguel R. Castellanos, and Robert W. Heath Jr. are
	 with the Department of Electrical and Computer Engineering, North Carolina State University, Raleigh, NC 27606 USA (email: \{nvdeshpa, mrcastel, rwheathjr\}@ncsu.edu).
	 Saeed R. Khosravirad, Jinfeng Du, and Harish Viswanathan are with Nokia
	 Bell Laboratories, Murray Hill, NJ 07974, USA (email: \{saeed.khosravirad,  jinfeng.du,
	 harish.viswanathan\}@nokia-bell-labs.com).
	  This project is funded by Nokia Bell Laboratories, Murray Hill, NJ 07974, USA.}
}

%\author{Miguel R. Castellanos, 
%        % <-this % stops a space
%\thanks{This paper was produced by the IEEE Publication Technology Group. They are in Piscataway, NJ.}% <-this % stops a space
%\thanks{Manuscript received April 19, 2021; revised August 16, 2021.}}



\maketitle

\begin{abstract}
The narrowband and far-field assumption in conventional wireless system design leads to a mismatch with the optimal beamforming required for wideband and near-field systems. This discrepancy is exacerbated for larger apertures and bandwidths. To characterize 
the behavior of near-field and wideband systems, we derive the beamforming gain expression achieved by a frequency-flat phased array designed for plane-wave propagation. To determine the far-field to near-field boundary for a wideband system, we propose a frequency-selective distance metric. The proposed far-field threshold increases for frequencies away from the center frequency. The analysis results in a fundamental upper bound on the product of the array aperture and the system bandwidth. We present numerical results to illustrate how the gain threshold affects the maximum usable bandwidth for the n260 and n261 5G NR bands.
\end{abstract}


\section{Introduction}
Distinguishing between the near-field and far-field is becoming increasingly relevant as modern wireless devices begin to operate in both propagation regions.
The most common near-field distance, known as the Fraunhofer   distance, is
proportional to the square of the aperture and inversely proportional to the  wavelength~\cite{selvan_fraunhofer_2017}. 
To satisfy the data rate requirements of 5G and beyond, wireless systems have shifted to higher carrier frequencies and larger antenna arrays\cite{tripathi_millimeter-wave_2021}\cite{heath_overview_2016}. 
For these modern arrays, the Fraunhofer distance becomes comparable   to the typical cell radius. For example, the Fraunhofer array distance for a uniform linear array (ULA) with 128 antennas and half-wavelength inter-antenna spacing operating at 28 GHz is around 88 m. This is a good fraction of the cell radius of  an urban microcellular and picocellular deployment~\cite{6732923}. 
In the near-field, the phase variation over the array aperture is non-linear in antenna index, which causes a phase mismatch  when assuming far-field propagation with  planar wavefronts~\cite{selvan_fraunhofer_2017}\cite{cui_channel_2022-1}. Inaccurate use of the far-field assumption can lead to beamforming gain losses that worsen as the array aperture increases\cite{cui_channel_2022-1}.

   Wideband systems with large arrays suffer from a phenomemon known as beam squint, i.e.,
  the mismatch between the frequency-flat response of the phased-array and the frequency-selective response of the wideband channel. This  reduces the beamforming  gain for frequencies away from the center frequency. For wideband systems operating in the near-field, the two mismatches due to far-field and narrowband system design jointly affect the beamforming  gain.
  Hence, it is crucial to  characterize this combined effect for the performance analysis of large phased-array wideband systems.

Near-field and wideband effects have generally been considered separately in the literature.
Various near-field metrics have been proposed in prior work\cite{bjornson_power_2020,cui_near-field_2021-1 }. The definition of the Fraunhofer array distance is based on the phase variation of a monochromatic wave over  the array length\cite{selvan_fraunhofer_2017}.
In \cite{bjornson_power_2020}, the proposed near-field metric
     used the amplitude variations as the criterion for determining the far-field to near-field transition distance.   
      The metrics in \cite{bjornson_power_2020,selvan_fraunhofer_2017} did not consider the angle of incidence  and have assumed broadside incidence only. The effective Rayleigh distance    incorporates the incidence angle in the  beamforming gain analysis\cite{cui_near-field_2021-1}. 
 One common shortcoming of these methods  is that the  beamforming gain analysis    is restricted only for a single frequency and not for a band of frequencies.
 Although \cite{7841766,gao_wideband_2021} analyzed the beamforming gain for wideband systems and incorporated beam squint effect, they are based on the plane-wave approximation.
To summarize, the existing work on the beamforming gain analysis either assumes a near-field and narrowband system \cite{bjornson_power_2020,selvan_fraunhofer_2017,cui_near-field_2021-1} 
or a far-field and wideband system\cite{7841766,gao_wideband_2021}.

In this letter, 
we analyze the beamforming gain of a   multiple-input-single-output (MISO) communication system with ULA at transmitter for a general near-field and wideband channel  with a beamformer based on the  far-field and narrowband assumption.
For a large number of antennas, the beamforming gain  can be approximated with a closed-form expression.
The derived expression can be generalized to arbitrary carrier frequencies and array sizes. 
  We propose the bandwidth-aware-near-field distance ($\mathsf{BAND}$) to characterize the far-field to near-field transition in a wideband system. The $\mathsf{BAND}$ increases for frequencies away from the center frequency, which implies that wideband systems have a larger near-field region.
 Expressing the system parameters as a function of the beamforming gain uncovers a  tradeoff  between the aperture and bandwidth.

\textit{Notation}: A bold lowercase letter $\ba$ denotes a vector,  $(\cdot)^{\ast}$ denotes conjugate transpose, $|\cdot|$  indicates absolute value, $[\ba]_n$ denotes the $n$th element of $\ba$, $\cO(\cdot)$ denotes the big Oh notation, 
$\mathcal{C}(\gamma)=\int_{0}^{\gamma} \cos(\frac{\pi}{2}t^2)\mbox{d}t $ and 
$\mathcal{S}(\gamma)=\int_{0}^{\gamma} \sin(\frac{\pi}{2}t^2)\mbox{d}t $ denote cosine and sine Fresnel functions, $\text{inf}\{\cdot\}$ denotes the infimum.

\section{System model}
Let us assume a {co-polarized} single-user MISO communication system  with an $N$ antenna ULA at the transmitter and a single antenna at the receiver. All antennas are assumed to be isotropic.
 The transmitter is oriented  along the $x$ axis with inter-antenna spacing $d$.
The 
$x$ coordinate of the $n$th transmit antenna is  defined as $ d_n=\frac{2n -N+1}{2}d$ for $ n=0,1, \dots, N-1$. The receiver is located at a distance $r$ from the transmit array center, i.e., the origin, and at an angle $\theta$ with the $y$ axis. 
The receiver location is $(r \sin(\theta), r \cos(\theta))$. 
The distance between the receive antenna and the $n$th transmit antenna is  $ r_n= \sqrt{r^2- 2 r d_n \sin(\theta)+ d_n^2}$.

The channel is assumed to be  a  line-of-sight (LOS) path between the transmitter and receiver. 
The path can be characterized by the path loss and the path delay.
The  path loss between all transmit antennas and the receive antenna is assumed to be the same and denoted by $G(r)$. This assumption is valid for $r\sim \cO(L)$ where $L=N d$ is the array aperture\cite{lu_communicating_2021}.
Let     $f_\mathsf{c}$ denote the   center frequency of the passband signal  and    $c$ denote the speed of light. 
The passband time-domain channel impulse response from the $n$th antenna is
\begin{equation}
 [\boldsymbol{\mathit{h}}_{\mathsf{p}}(t)]_n= \sqrt{G(r)} \delta\left(t-\frac{r_n}{c}\right).
\end{equation}
 The pseudo-complex baseband equivalent channel response  after down-conversion is 
 \begin{equation}
  [\boldsymbol{\mathit{h}}_{\mathsf{b}}(t)]_n= \sqrt{G(r)} e^{-j 2\pi \frac{r_n}{c}f_\mathsf{c} }\delta\left(t-\frac{r_n}{c}\right).
 \end{equation}
The frequency-domain channel impulse response at baseband frequency  $f$ is 
\begin{equation}\label{eqn: f domain channel impulse}
[\boldsymbol{\mathsf{h}}(f)]_n= \sqrt{G(r)} e^{-j 2\pi \frac{r_n}{c}(f_\mathsf{c} +f)} .
\end{equation}
The general channel response in \eqref{eqn: f domain channel impulse} can be approximated under the narrowband and far-field assumptions. Under the narrowband assumption, the baseband frequency can be treated as small, i.e. $f  \approx 0$. 
Taking the series expansion of the channel phase response around $f = 0$, we have $\exp({-j 2\pi \frac{r_n}{c}(f_\mathsf{c} +f)}) = \exp({-j 2\pi \frac{r_n}{c}(f_\mathsf{c}  + \cO(f))}) $. Under the far-field assumption,
the distance $r$ can be treated large, i.e., $r \rightarrow \infty$.
Taking the series expansion of the channel phase response around $r \rightarrow \infty$,
 we have   $\exp({-j 2\pi \frac{r_n}{c}(f_\mathsf{c} +f)  })=  \exp\left({-j 2\pi \frac{\left(r - d_n \sin(\theta)+ \cO\left(\frac{1}{r}\right)\right)}{c}(f_\mathsf{c} +f)  } \right)$. Combining both of the expansions, we have,  $\exp\left({-j 2\pi \frac{r_n}{c}(f_\mathsf{c} +f)  }\right) = \\ \exp\left(-j 2\pi \frac{\left(r - d_n \sin(\theta)+ \cO\left(\frac{1}{r}\right)\right)}{c} (f_\mathsf{c}  + \cO(f))  \right) $ when the narrowband and far-field assumptions both hold.
Using the subscripts $\mathsf{nf}$, $\mathsf{ff}$,  $\mathsf{wb}$, and $\mathsf{nb}$ to denote near-field, far-field, wideband, and narrowband assumptions, respectively, we summarize the 
four  channel models, for the $N \times 1$ channel vectors, $\boldsymbol{\mathsf{h}}_{\mathsf{nf,wb}}(f)$, $\boldsymbol{\mathsf{h}}_{\mathsf{nf,nb}}$, $\boldsymbol{\mathsf{h}}_{\mathsf{ff,wb}}(f)$, and $\boldsymbol{\mathsf{h}}_{\mathsf{ff,nb}}$,   in Table \ref{tab: summary channels}.
   
\begin{table}[htbp]
	\centering
	\caption{Summary of the channel models}
	\label{tab: summary channels}
	
	\begin{tabular}{|c |c |} 
		\hline
		Channel response &  Expression \\ [0.5ex] 
		\hline
		$[\boldsymbol{\mathsf{h}}_{\mathsf{nf,wb}}(f)]_{n}$  & 	$ \sqrt{G(r)} e^{-j 2\pi \frac{r_n}{c}(f_\mathsf{c} +f)  }$ \\ 
		\hline
		$[\boldsymbol{\mathsf{h}}_{\mathsf{nf,nb}}]_{n}$ & $	 \sqrt{G(r)}  e^{-j 2\pi \frac{r_n}{c} f_{\mathsf{c}} }$ \\
		\hline
		$[\boldsymbol{\mathsf{h}}_{\mathsf{ff,wb}}(f)]_{n}$ & $	 \sqrt{G(r)} e^{-j 2\pi \frac{(r - d_n \sin(\theta))}{c}(f_\mathsf{c} +f)  }$ \\
		\hline
		$[\boldsymbol{\mathsf{h}}_{\mathsf{ff,nb}}]_{n}$ & $\sqrt{G(r)} e^{-j 2\pi \frac{(r - d_n \sin(\theta))}{c}f_\mathsf{c}   }  $ \\ [1ex] 
		\hline
	\end{tabular}
	
\end{table}

The most general channel response is $\boldsymbol{\mathsf{h}}_{\mathsf{nf,wb}}(f) $. 
%For a complex sinusoidal signal $x$ transmitted at   $f$, assuming beamforming $\bff$ at the transmitter and additive noise $z$, the received signal is 
%\begin{equation}
%	y= \boldsymbol{\mathsf{h}}_{\mathsf{nf,wb}}^{\ast}(f) \bff x+z.
%\end{equation}
The optimal beamforming vector with unit norm constraint that maximizes the signal to noise ratio is  $\bff_{\mathsf{nf,wb}}(f)=\frac{1}{\sqrt{N G(r)}} \boldsymbol{\mathsf{h}}_{\mathsf{nf,wb}}(f)  $. Using a mismatched beamforming vector, i.e., $\bff  \neq \bff_{\mathsf{nf,wb}}(f)$ leads to beamforming gain loss.
To characterize this mismatch, we  define 
 vectors  $\bff_{\mathsf{nf,nb}}$, $\bff_{\mathsf{ff,wb}}(f)$, and $\bff_{\mathsf{ff,nb}} $ similarly.

Each of the four beamforming vectors allow different  hardware implementations. 
 The beamforming vectors $\bff_{\mathsf{ff,nb}} $  and $\bff_{\mathsf{nf,nb}}$ are frequency-flat and can be implemented using narrowband phase-shifters, whereas  $\bff_{\mathsf{nf,wb}}(f)  $ and $\bff_{\mathsf{ff,wb}}(f)$ are frequency-selective and can be implemented at higher cost as a fully-digital space-time precoder \cite{daniels2005miso} or an analog true-time-delay  architecture\cite{rotman_true_2016}. 
 In this letter, we analyze the 
 mismatch that occurs when using $\bff_{\mathsf{ff,nb}}$ for a near-field wideband system. This analysis is crucial to characterize the scenarios where  frequency-flat beamforming  does not work well.

\section{Beamforming gain analysis}\label{sec: array gain analysis}    
  %For performance analysis, we use the array gain metric.
  We analyze the  performance loss due to  $\bff_{\mathsf{ff,nb}}$ 
   in terms of the  normalized beamforming gain. The definitions of the normalized beamforming gains under different channel assumptions are summarized in Table \ref{tab: norm array gain}. The normalization is such that the maximum gain value is 0 dB.
   
   \begin{table}[htbp]
   	\centering
   	\caption{Definitions of normalized beamforming gains}
   	
   	\label{tab: norm array gain}
   	\begin{tabular}{|c |c |} 
   		\hline
   		Normalized beamforming gain &  Expression \\ [0.25ex]  
   		\hline
   		$	\mu_{\mathsf{nf,wb}}(f)$  & 	$ \frac{1}{\sqrt{G(r) N}} \left| \boldsymbol{\mathsf{h}}_{\mathsf{nf,wb}}^{\ast}(f)\bff_{\mathsf{ff,nb}} \right|$ \\  [1.5ex] 
   	\hline
   		$\mu_{\mathsf{ff,wb}}(f)$\cite{gao_wideband_2021} & $ \frac{1}{\sqrt{G(r) N}}  \left|\boldsymbol{\mathsf{h}}_{\mathsf{ff,wb}}^{\ast}(f)\bff_{\mathsf{ff,nb}}   \right|$ \\
   		[1.5ex] 
   		\hline
   	
   		$\mu_{\mathsf{nf,nb}}$\cite{cui_channel_2022-1} &  $\frac{1}{\sqrt{G(r) N}}  \left| \boldsymbol{\mathsf{h}}_{\mathsf{nf,nb}}^{\ast}\bff_{\mathsf{ff,nb}}  \right|$ \\
   	 [1.5ex] 
   		\hline
   	\end{tabular}
   
   \end{table}
The beamforming gain $\mu_{\mathsf{ff,wb}}(f)$ only captures the wideband effect  and   $	\mu_{\mathsf{nf,nb}}$ only captures the near-field effect.  

In Lemma~\ref{lem: mu r theta f}, we express the inner product in $	\mu_{\mathsf{nf,wb}}(f)$  as a summation of the complex phases over $N$ antennas. To express the phase in terms of dimensionless parameters, 
let the normalized distance be $\bar{r}=\frac{r}{\lambda_\mathsf{c}}$, the normalized inter-antenna spacing be $ \bar{d}=\frac{d}{\lambda_\mathsf{c}}$, and the
normalized frequency be  $\bar{f}=\frac{f}{f_\mathsf{c}}$.
In the radiating near-field (Fresnel) region, i.e.,  $r_n > 0.5\sqrt{\frac{L^3}{\lambda_\mathsf{c}}}$,  the $\cO{\left(\frac{1}{r^2}\right)}$ term in the expansion of $r_n=  r- d_n \sin(\theta) + \frac{ d_n^2 \cos^2(\theta)}{2 r}+ \cO{\left(\frac{1}{r^2}\right)}$ can be ignored~\cite{selvan_fraunhofer_2017}.
%We use the series expansion of $r_n$ around $\infty$ to obtain \cite{cui_channel_2022-1}
%\begin{equation}\label{eqn: rn series expansion}
%r_n=  r- d_n \sin(\theta) + \frac{ d_n^2 \cos^2(\theta)}{2 r}+ \cO{\left(\frac{1}{r^2}\right)}.
%\end{equation}
%Dropping the term $\cO{\left(\frac{1}{r^2}\right)}$ in the expansion of $r_n$ in \eqref{eqn: rn series expansion} is valid 
% under the assumption $r_n > 0.5\sqrt{\frac{N^3 d^3}{\lambda_\mathsf{c}}}$\cite{cui_channel_2022-1}. 
 %In Section \ref{sec: results}, we will justify this assumption.
 This assumption enables the decomposition of the  phase for each antenna into two phase terms that individually capture the  near-field and wideband phenomenon.
 At the $n$th antenna, we define the  phase contribution due to the wideband assumption, $\phi_{\mathsf{wb}}=-n \bar{d} \sin(\theta) \bar{f} $, and the phase contribution due to the near-field assumption, $ \phi_{\mathsf{nf}}=(\bar{f}+1) \frac{\bar{d}^2}{2 \bar{r}}\cos^2(\theta) \left(n- \frac{N-1}{2}\right)^2   $.
%\blue{The residual phase term $\widetilde{\phi}$ is  $\cO(\frac{1}{\bar{r}^2})$.}

\begin{lemma}\label{lem: mu r theta f}
In the Fresnel region, 	the near-field wideband beamforming gain can be approximated  as $\mu_{\mathsf{nf,wb}}^{\mathsf{approx}}(\bar{f})$, where
\begin{equation}\label{eqn: mu r theta f summation form}
\mu_{\mathsf{nf,wb}}^{\mathsf{approx}}(\bar{f})= \left|\frac{1}{N}  \sum_{n=0}^{N-1}  e^{j{2\pi}(\phi_{\mathsf{wb}}+\phi_{\mathsf{nf}}  )}\right|  .
\end{equation}
	\end{lemma}
\begin{proof}
The phase term in the inner product $\boldsymbol{\mathsf{h}}_{\mathsf{nf,wb}}^{\ast}(f)\bff_{\mathsf{ff,nb}} $ at the $n$th antenna  is
\begin{align}
	&=-\angle([\boldsymbol{\mathsf{h}}_{\mathsf{nf,wb}}(f)]_n)+ 	\angle( [\bff_{\mathsf{ff,nb}}]_n)  \nonumber \\
	&  \stackrel{(a)}{=} - \angle([\boldsymbol{\mathsf{h}}_{\mathsf{nf,wb}}(f)]_n)+  \angle([\boldsymbol{\mathsf{h}}_{\mathsf{ff,nb}}]_n)  \nonumber \\
	 & \stackrel{(b)}{\approx}\frac{2\pi}{c}\bigg[ \left(r- d_n \sin(\theta) + \frac{ d_n^2 \cos^2(\theta)}{2 r} \right)(f_\mathsf{c} +f)
	 \nonumber \\ &-(r- d_n \sin(\theta))f_\mathsf{c}\bigg],
\end{align}
 where equality $(a)$
follows from the definition of $\bff_{\mathsf{ff,nb}}$ and approximation $(b)$ follows from Table ~\ref{tab: summary channels} and expansion of $r_n$. By factoring out the terms which do not depend on the index $n$, and using the definition of $\bar{f}$, $\bar{d}$, and $\bar{r}$, we get \eqref{eqn: mu r theta f summation form}.
 \end{proof}

Lemma~\ref{lem: mu r theta f} illustrates a few key  insights.
The expression for $\mu_{\mathsf{nf,wb}}(f)$ depends only on the normalized parameters. Hence, this analysis is independent of the carrier frequency; the same beamforming gain can be obtained for different carrier frequencies provided that the normalized parameters remain fixed. 
The expression for the  far-field wideband  gain $		\mu_{\mathsf{ff,wb}}(f)$ in \cite{gao_wideband_2021} can be obtained by setting $ \phi_{\mathsf{nf}}=0$.
The near-field narrowband gain $		\mu_{\mathsf{nf,nb}}$ in
\cite{cui_channel_2022-1} can be obtained by 
 setting $\phi_{\mathsf{wb}}=0 $ and $\bar{f}=0$.
The phase
$\phi_{\mathsf{wb}} $ is linear and   $ \phi_{\mathsf{nf}}$  is quadratic in the antenna index. Both $\phi_{\mathsf{wb}} $ and $ \phi_{\mathsf{nf}}$ are dependent on $\bar{f}$. Hence, existing beamforming methods based on uniform spacing and narrowband assumption do not work well.

In Lemma~\ref{lem: mu r theta f fresnel}, we  further simplify  \eqref{eqn: mu r theta f summation form} to get an expression  in terms of Fresnel functions whose arguments depend on the system parameters. 
 The purpose of this simplification is to establish a closed-form algebraic relationship between the system parameters and the beamforming gain threshold.
We define the normalized  array aperture as $\bar{L}= N \bar{d}$.
Lemma~\ref{lem: mu r theta f fresnel} shows the relationship between the parameters, $\{ \bar{f}, \bar{r}, \bar{L}, \theta \}$, and a 2D parameter space defined
by $\gamma_1$ and $\gamma_2$ as 
\begin{equation}\label{eqn: gamma1}
	\gamma_1=-\tan(\theta) \bar{f} \sqrt{\frac{  2\bar{r}}{ {1+\bar{f}}}},
\end{equation} 
\begin{equation}\label{eqn: gamma2}
	\gamma_2=   \bar{L}\cos(\theta)\sqrt{ \frac{1+\bar{f}}{2\bar{r}}}.
\end{equation}
The compression from four  parameters to two new parameters $\gamma_1$ and $\gamma_2$ allows us to visualize the beamforming gain function. It also simplifies numerical simulations by varying $\gamma_1$ and $\gamma_2$ instead of varying four system parameters.

%\begin{figure}[htbp]
%	\includegraphics[width=0.55\textwidth]{fixed_gamma2_vary_gamma1_2d_plot.eps}
%	\caption{2D cross-section of $\cG(\gamma_1, \gamma_2) $ vs $\gamma_1$ for fixed $\gamma_2$. }
%	\label{fig:2d plot fixed gamma2}
%\end{figure}
{\begin{lemma}\label{lem: mu r theta f fresnel}
		 The expression of $\mu_{\mathsf{nf,wb}}^{\mathsf{approx}}(\bar{f})$  in \eqref{eqn: mu r theta f summation form} can be further approximated for large $N$ with fixed $d$ 
	 as
	\begin{equation}\label{eqn: G gamma1 gamma2}
		\cG(\gamma_1, \gamma_2) = \underset{N\rightarrow \infty}{\mathsf{lim}} \mu_{\mathsf{nf,wb}}^{\mathsf{approx}}(\bar{f})= \left|\frac{ \bar{\mathcal{C}}(\gamma_1, \gamma_2)+ j\bar{\mathcal{S}}(\gamma_1, \gamma_2)}{2 \gamma_2}  \right|,
	\end{equation}
$\bar{\mathcal{C}}(\gamma_1, \gamma_2)\equiv \mathcal{C}(\gamma_1+\gamma_2)-\mathcal{C}(\gamma_1-\gamma_2) $ and $\bar{\mathcal{S}}(\gamma_1, \gamma_2)\equiv \mathcal{S}(\gamma_1+\gamma_2)-\mathcal{S}(\gamma_1-\gamma_2) $. 
\end{lemma}}
\begin{proof}
{We	rewrite \eqref{eqn: mu r theta f summation form} by defining $\Delta_m=\frac{1}{N}$ for $m=0, \frac{1}{N}, \dots, \frac{N-1}{N}$,  $a=\cos(\theta)\sqrt{\frac{(1+\bar{f}) \bar{d}^2}{2 \bar{r}}}$ , and $b=\frac{1}{a}\left({\frac{(1+\bar{f})\bar{d}^2 \cos^2(\theta) (N-1)}{4 \bar{r}}+ \frac{\bar{d} \sin(\theta) \bar{f}}{2}}\right)$ to get 
	\begin{equation}\label{eqn: app, mu r theta f change of variables}
		\mu_{\mathsf{nf,wb}}^{\mathsf{approx}}(\bar{f}) =   \left|\Delta_m  \sum_{m=0}^{1-\frac{1}{N}} \exp\left(j{2\pi}( a m N -b )^2\right)\right|.
	\end{equation}
	As $N \rightarrow \infty$, we can express the summation in \eqref{eqn: app, mu r theta f change of variables} as an integral using the Riemann integral method\cite{cui_channel_2022-1}  as
	\begin{align}\label{eqn: mu nfwb approx}
		\mu_{\mathsf{nf,wb}}^{\mathsf{approx}}(\bar{f})   &=  \left| \int_0^{1}  \exp\left(j{2\pi}( a  N t -b )^2\right) dt     + \cO\left(\frac{1}{N}\right)\right|, \nonumber \\
		& \equalityA \left|\frac{ \int_{-2b}^{2 a N- 2b}  \exp\left(j\frac{\pi}{2}( t^{\prime} )^2\right) dt^{\prime}}{2 a N}  + \cO\left(\frac{1}{N}\right)\right|,
	\end{align}
	where $(a) $ follows by letting $a N t- b= \frac{t^{\prime}}{2}$. For large $N$, using  $(N-1)\bar{d}\approx (N+1)\bar{d}\approx N \bar{d}=\bar{L}$, \eqref{eqn: gamma1}, \eqref{eqn: gamma2},  we get  \eqref{eqn: G gamma1 gamma2}.}
\end{proof}

The expression derived in Lemma \ref{lem: mu r theta f fresnel} simplifies to the narrowband case, i.e., $\bar{f}=0$, by substituting $\gamma_1=0$ in \eqref{eqn: G gamma1 gamma2}. Hence, the near-field narrowband gain for large $N$ is defined as $\cG^{\mathsf{nb}}(\gamma_2) =\underset{N\rightarrow \infty}{\mathsf{lim}} \mu_{\mathsf{nf,nb}}= \left|\frac{ {\mathcal{C}}( \gamma_2)+ j{\mathcal{S}}(\gamma_2)}{\gamma_2}\right|$, which is the same near-field gain expression derived in \cite{cui_channel_2022-1}  and \cite{ cui_near-field_2021-1}.
We have generalized $\cG^{\mathsf{nb}}(\gamma_2)$ \cite{cui_channel_2022-1}, \cite{ cui_near-field_2021-1} by incorporating the bandwidth effect through the parameter $\gamma_1$. 
%\begin{corollary}
%	  The expression derived in Lemma \ref{lem: mu r theta f fresnel} can be simplified to the narrowband case i.e., $\bar{f}=1$ by substituting $\gamma_1=0$ in \eqref{eqn: G gamma1 gamma2}. Hence, $\mu(r, \theta)\approx \cG^{nb}(\gamma_2)=\left|\frac{ {\mathcal{C}}( \gamma_2)+ j{\mathcal{S}}(\gamma_2)}{\gamma_2}\right|$, which is the same expression derived in \cite{cui_channel_2022-1}  and \cite{ cui_near-field_2021-1}.
%\end{corollary}
\begin{figure}[htbp]
	
	\centering
	\includegraphics[width=0.4\textwidth]{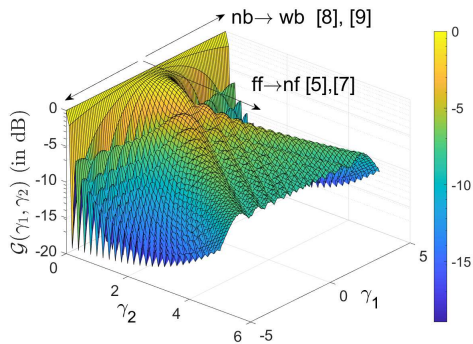}
	
	\caption{3D plot of $\cG(\gamma_1, \gamma_2) $ vs $\gamma_1$ and $\gamma_2$ in dB scale. Increase in $|\gamma_1|$ corresponds to transition from narrowband to wideband; increase in $\gamma_2$ corresponds to transition from far-field to near-field. }
	\label{fig:3d plot gamma1 gamma2}
\end{figure}

\begin{figure}[htbp]
	\centering
	\includegraphics[width=0.45\textwidth]{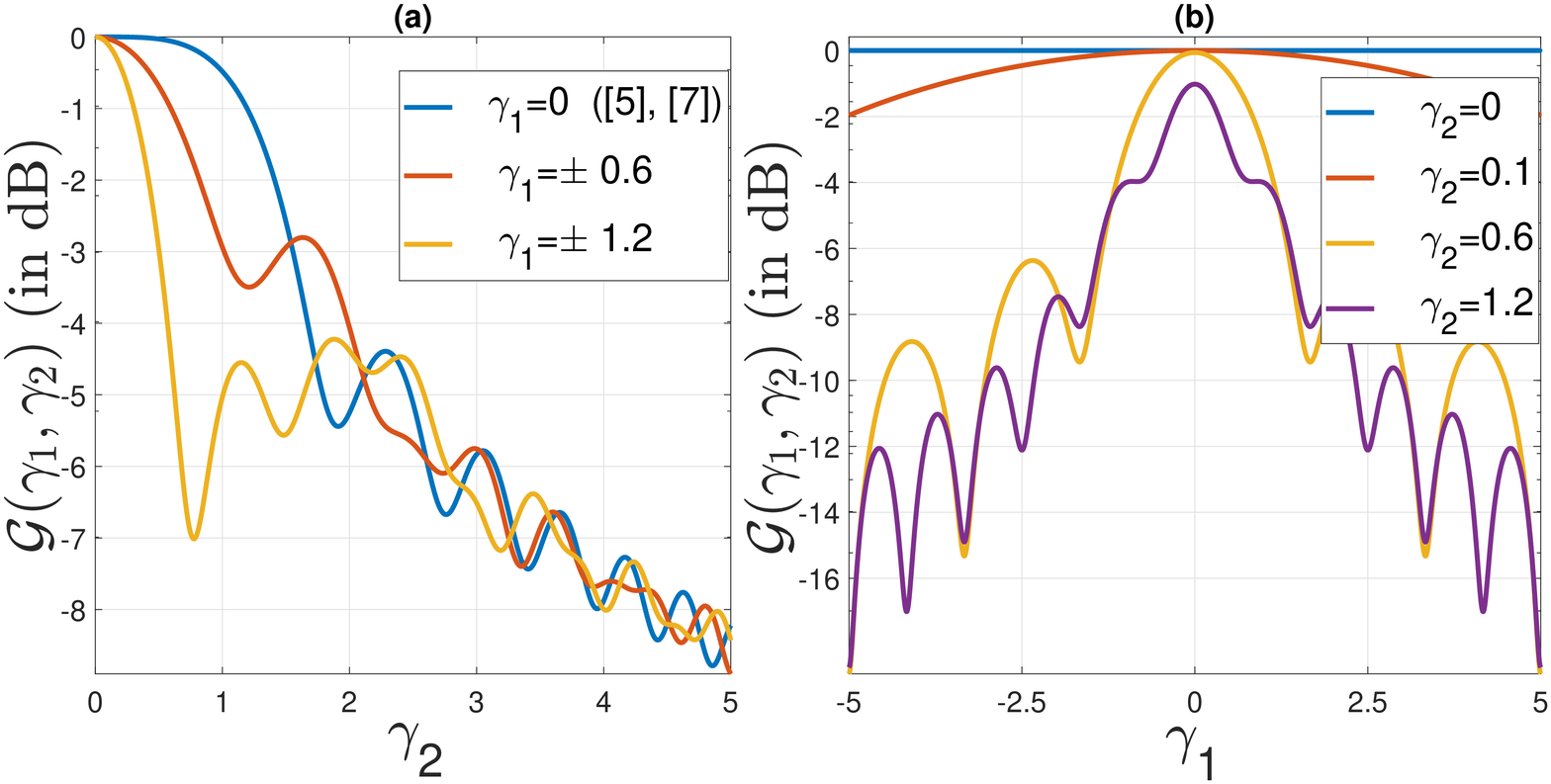}
	
	\caption{2D cross-sections of $\cG(\gamma_1, \gamma_2) $ in dB scale. (a) $\cG(\gamma_1, \gamma_2) $ vs $\gamma_2$ for a fixed  $|\gamma_1|$; (b) $\cG(\gamma_1, \gamma_2) $ vs $\gamma_1$ for a fixed $\gamma_2$.  }
	\label{fig:2d plot crosssections}
\end{figure}

In Fig.~\ref{fig:3d plot gamma1 gamma2}, we study the variation of $\cG(\gamma_1, \gamma_2) $ with $\gamma_1$ and $\gamma_2$ to understand the beamforming gain dependence on the system parameters.
There is a drastic reduction in $\cG(\gamma_1, \gamma_2) $ with increase in $| \gamma_1|$ and $\gamma_2$. Note that the function $\cG(\gamma_1, \gamma_2) $ is an even function with respect to $\gamma_1$. This reflects the fact that the beamforming gain  is symmetric with respect to $\theta$.

The 2D cross-sections of the beamforming gain are shown in Fig.~\ref{fig:2d plot crosssections}. 
%We discuss the significance of the parameters $\gamma_1$ and $\gamma_2$ by interpreting the 2D cross-sections of the 3D plot as shown in Fig.~\ref{fig:2d plot crosssections}. 
%\textbf{Key observations from Fig.~\ref{fig:2d plot crosssections}:} 
In Fig.~\ref{fig:2d plot crosssections} (a), we  keep $\gamma_1$ fixed and plot $\cG(\gamma_1, \gamma_2)$ as a function of $\gamma_2$. Assuming a non-zero angle $\theta$, the 2D plot corresponding to $\gamma_1=0$, shown in blue, represents the narrowband case which is same as the plot shown in prior work on near-field propagation \cite{cui_channel_2022-1}, \cite{ cui_near-field_2021-1}.
% For $\gamma_1=0$ and $\bar{f}=1$, the expression in \eqref{eqn: G gamma1 gamma2} can be written as a function of $\gamma_2$ as $\cG( \gamma_2)= \left|\frac{ {\mathcal{C}}( \gamma_2)+ j{\mathcal{S}}(\gamma_2)}{\gamma_2}\right|, $ which is the same expression derived in \cite{cui_channel_2022-1}  and \cite{ cui_near-field_2021-1}. 
 The parameter $\gamma_2$    captures the transition from far-field to near-field in \cite{ cui_near-field_2021-1}. The effective Rayleigh distance, $d_\mathsf{ERD}(\theta)\!=\! 0.367 \cos^2(\theta)({2 \bar{L}^2 \lambda_\mathsf{c} })$, is defined  using \eqref{eqn: gamma2} with $\bar{f}=0$, as the distance below which the value of the beamforming gain $\cG^{\mathsf{nb}}( \gamma_2)$ falls under the threshold 0.95 in linear scale \cite{ cui_near-field_2021-1}. 
% Our expression for $\cG(\gamma_1, \gamma_2) $ in \eqref{eqn: G gamma1 gamma2} is a generalization of the expression $\cG^{\mathsf{nb}}(\gamma_2)$ derived in \cite{cui_channel_2022-1}    and \cite{ cui_near-field_2021-1} by incorporating the bandwidth effect through the parameter $\gamma_1$.  
   We observe that for the  wideband case, i.e., $\bar{f}\neq 0$, the value of $\cG(\gamma_1, \gamma_2) $ drops sharply with $\gamma_2$ as $| \gamma_1|$ increases.
 % In fig.~\ref{fig:2d plot crosssections} (a), we mark the 0.95 threshold line and denote its intersection with the yellow, red, and blue plots by $\gamma_2^y, \gamma_2^r,$  and $\gamma_2^b$ respectively.
%We observe that
%$\gamma_2^y< \gamma_2^r <\gamma_2^b$.
%From \eqref{eqn: gamma2}, we have the relation that 
%$\gamma_2 \propto \frac{1}{\sqrt{\bar{r}}}$.
%Hence, keeping other parameters fixed, $\bar{r}^y> \bar{r}^r >\bar{r}^b=d_{\mathsf{ERD}}$. Hence, with increasing $\gamma_1$, the value of 
%   Hence, from \eqref{eqn: gamma2}, we can conclude that the value of the transition distance $\bar{r}$ will be higher than $d_{\mathsf{ERD}}$ provided other parameters remain fixed. 
In Fig.~\ref{fig:2d plot crosssections} (b),  $\gamma_2$ is fixed and we plot $\cG(\gamma_1, \gamma_2)$ as a function of $\gamma_1$.   For small values of $\gamma_2$, the peak value is close to  0 dB.
%The main lobe of the plot of $\gamma_2=0.6$ closely matches with $\Xi_4$ where $\Xi_N=|\frac{\sin(N x/2)}{N\sin(x/2)}|$ is the Dirichlet Sinc function conventionally used to analyze effect of beam squint under far-field assumption  \cite{7841766,gao_wideband_2021}. 
However, for larger values of $\gamma_2$,  the peak value drops and the main lobe shrinks. These results suggest that the joint  effect of  $\gamma_1$ and $\gamma_2$ is more severe than their individual effect.

The 2D cross-sections of the 3D plot in Fig.~\ref{fig:3d plot gamma1 gamma2}  resemble the plots from existing works \cite{cui_near-field_2021-1, 7841766,gao_wideband_2021}  which study the wideband and near-field phenomena separately.  The 3D plot jointly models the wideband and near-field effects.
%As we increase the value of $\gamma_2$ to 0.2 and 0.6, we observe that the bandwidth plays a role through the parameter $\gamma_1$.
 We analyze the connections of $\{ \gamma_1, \gamma_2\}$ with $\{ \bar{f}, \bar{r}, \bar{L}, \theta \}$ in Section \ref{sec: inverse mapping}.

\section{Inverse mapping of beamforming gain to system parameters space}\label{sec: inverse mapping}
Most of the existing  studies \cite{bjornson_power_2020,selvan_fraunhofer_2017,cui_near-field_2021-1, 7841766,gao_wideband_2021} analyze the beamforming gain as a function of the different system parameters. From a system design perspective, however, it is essential to understand the inverse relationship for each system parameter as a function of the beamforming gain and other system parameters. We identify a fundamental tradeoff between the  aperture and bandwidth in Section \ref{subsec: ABP}. We also establish a frequency-selective near-field boundary distance in Section~\ref{subsec: BAND}. 

\subsection{Aperture-bandwidth product }\label{subsec: ABP}

In Section \ref{sec: array gain analysis}, we introduced the beamforming gain dependence on $\bar{f}$ through $\gamma_1$ and  $\gamma_2$. 
From \eqref{eqn: gamma1} and \eqref{eqn: gamma2},  the normalized frequency $\bar{f}$ can be expressed as 
\begin{equation}\label{eqn: bar f}
	\bar{f}= - \frac{\gamma_1 \gamma_2}{\bar{L} \sin(\theta)}.
\end{equation}
We also define the fractional bandwidth 
as $f_{\mathsf{B}}=|2\bar{f}|$. Hence, from \eqref{eqn: bar f}, we have the relation
\begin{equation}\label{eqn: fbw l sin theta gamma 1 gamma 2}
|	f_{\mathsf{B}} \bar{L} \sin(\theta)|=|2 \gamma_1 \gamma_2|.
\end{equation}
To understand the maximum limit up to which the aperture and/or bandwidth can be scaled up while maintaining the narrowband and far-field assumption, we are interested in the maximum limit of the right hand side of  \eqref{eqn: fbw l sin theta gamma 1 gamma 2} which can be found numerically for a given value of $\cG(\gamma_1, \gamma_2)$.

We show the 2D contour plot version of Fig.~\ref{fig:3d plot gamma1 gamma2} in Fig.~\ref{fig:2d plot gamma1 gamma2 product}.
%In Fig.~\ref{fig:2d plot gamma1 gamma2 product}, we show the 2D contour plot version of Fig.~\ref{fig:3d plot gamma1 gamma2}. 
%The contour plot describes the regions in the $(\gamma_1, \gamma_2)$ space which result in array gain above a fixed threshold.
The contour plot is defined as the locus of the points in the $(\gamma_1, \gamma_2)$ space which achieve a fixed value of the beamforming gain $\cG(\gamma_1, \gamma_2)$. The pair of hyperbola marked in red, $ |\gamma_1 \gamma_2|= 0.5044$, corresponds to $\cG(\gamma_1, \gamma_2)= -2 $ dB. 
The pair of hyperbola marked in blue, $ |\gamma_1 \gamma_2|= 0.3654$, corresponds to $\cG(\gamma_1, \gamma_2)= -1 $ dB. 
Hence, from \eqref{eqn: fbw l sin theta gamma 1 gamma 2} and Fig.~\ref{fig:2d plot gamma1 gamma2 product}, we conclude that to maintain a beamforming gain, $\cG(\gamma_1, \gamma_2)\in [-2, -1]$ dB, $|f_{\mathsf{B}} \bar{L} \sin(\theta)| $ must approximately lie in the range $[0.73, 1]$.
The  upper limit on  $	|f_{\mathsf{B}}  \bar{L} \sin(\theta)|$  decreases as the threshold increases. From a system design perspective, the worst case angle of incidence $\theta_{\mathsf{worst}}$ can be chosen based on the sector division.
%Let $\theta_{\mathsf{max}}$ denote the maximum allowable angular span.
 Hence,  we get the worst case upper bound on $	f_{\mathsf{B}}  \bar{L}$ as $f_{\mathsf{B}}  \bar{L} \leq \left|\frac{ 2 \gamma_1 \gamma_2}{\sin(\theta_{\mathsf{worst}} )}\right|.$
In terms of un-normalized parameters, this simplifies to an important fundamental constraint on the product of the bandwidth, $B $, and the array aperture, $L$, as
\begin{equation}\label{eqn: upper bound abp}
	B   L \leq  \left|\frac{ 2  c \gamma_1 \gamma_2}{\sin(\theta_{\mathsf{worst}}  )}\right|.
\end{equation}
  The relationship in \eqref{eqn: upper bound abp} plays an important role in determining the limits on the system design parameters for a particular beamforming gain. We define the maximum usable bandwidth, for a fixed aperture $L$, that achieves beamforming gain $\tau$,  for the worst case incidence angle, as $B_{\mathsf{max}}= \left|\frac{ 2  c [\gamma_1 \gamma_2]_{\mathsf{max}}}{L \sin(\theta_{\mathsf{worst}} )} \right|$, where $[\gamma_1 \gamma_2]_{\mathsf{max}}$ is computed numerically for a given $\tau$. If the system bandwidth exceeds $B_{\mathsf{max}}$, the beamforming gain drops below the required threshold.
  \begin{figure}[htbp]
  	
  	\centering
  	\includegraphics[width=0.5\textwidth]{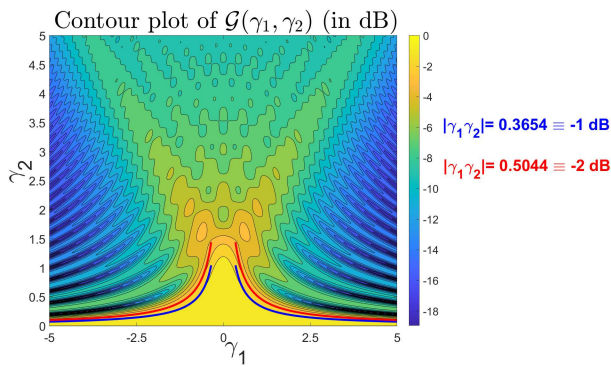}
  	\caption{Contour plot of $\cG(\gamma_1, \gamma_2) $ in dB scale. The hyperbolas determine  upper bound on the aperture-bandwidth product for a fixed beamforming gain.}
  	
  	\label{fig:2d plot gamma1 gamma2 product}
  \end{figure}
\begin{remark}
	 We observe that $f_{\mathsf{B}}  \bar{L} $ can also be written as the ratio of the maximum propagation delay difference across the array to the symbol duration $T_{\mathsf{s}}$, i.e.,  $f_{\mathsf{B}}  \bar{L} =\frac{N d/c}{T_{\mathsf{s}}} $. In \cite{jang_1-ghz_2019}, the upper bound on $f_{\mathsf{B}}  \bar{L} $  was loosely specified using $N d/c \ll T_{\mathsf{s}} $.
	   The bound proposed in \eqref{eqn: upper bound abp} is more precise.
\end{remark}

\subsection{Bandwidth-aware-near-field distance (BAND)}\label{subsec: BAND}
We use the relationship in \eqref{eqn: gamma2}  to derive a frequency-dependent near-field distance.
The normalized distance, $\bar{r}$,   can be written as  a function of $\bar{f}$, $\gamma_2$, $\bar{L}$, and $\theta$ as 
\begin{equation}\label{eqn: bar r}
\bar{r}(\bar{f}, \gamma_2, \bar{L}, \theta)= \frac{\bar{L}^2 \cos^2(\theta)(1+\bar{f})}{2 \gamma_2^2}.
%=\frac{\bar{L}^2 \cos^2(\theta)\left(1-  \frac{\gamma_1 \gamma_2}{\bar{L} \sin(\theta)}\right)}{2 \gamma_2^2}.
\end{equation}
%The un-normalized distance can be computed from \eqref{eqn: bar r} as $ r= \lambda_\mathsf{c}\bar{r}(\gamma_1, \gamma_2)$.
The $\mathsf{BAND}$ is the smallest distance  beyond which the beamforming gain  is always above a certain threshold $\tau$. 
 By expressing $\gamma_1$ using $\{\bar{f}, \gamma_2, \bar{L}, \theta\}$,
  the $\mathsf{BAND}$  is defined as
\begin{align}
\mathsf{BAND}(f,f_{\mathsf{c}}, \tau, {L}, \theta)\!	= & \text{inf}\bigg\{\! \lambda_\mathsf{c}	\bar{r}^{\prime}:\! \cG\left(\!\frac{-\bar{f}\bar{L}\sin(\theta)}{\gamma_2}, \gamma_2\! \right)\geq \tau,\nonumber\\ 
	&  \forall  \bar{r}(\bar{f}, \gamma_2, \bar{L}, \theta) \geq  	\bar{r}^{\prime}\bigg\}. 
\end{align}
The $\mathsf{BAND}$ is computed using  \eqref{eqn: bar r}  from  $\left(\frac{-\bar{f}\bar{L}\sin(\theta)}{\gamma_2}, \gamma_2\right)$,  which is obtained numerically for a  $\tau$  from Fig.~\ref{fig:2d plot gamma1 gamma2 product}.

 The $\mathsf{BAND}$ is related to the distances defined only for the narrowband case, i.e., $\bar{f}=0$, as follows:
\begin{itemize}
	\item Effective Rayleigh distance \cite{ cui_near-field_2021-1}: $d_{\mathsf{ERD}}(\theta)=\mathsf{BAND}(0, f_{\mathsf{c}},  -0.2 \mathsf{ dB},  {L}, \theta )$.
	\item Fraunhofer array distance \cite{selvan_fraunhofer_2017}: $d_{\mathsf{FA}}= 2  \bar{L}^2\lambda_\mathsf{c}= \mathsf{BAND}(0, f_{\mathsf{c}}, -0.04 \mathsf{ dB},  {L}, 0^{\circ} )  $.
	%\item Bjornson distance, $d_B= $
\end{itemize}
The $\mathsf{BAND}$ is a wideband generalization of the near-field distances proposed in \cite{selvan_fraunhofer_2017} and \cite{ cui_near-field_2021-1}. The $\mathsf{BAND}$ also determines the favorable regime for a transceiver hardware operating at any general frequency  offset from $f_{\mathsf{c}}$.

\section{Numerical results}\label{sec: results}

The analysis presented in Section \ref{sec: array gain analysis} and Section \ref{sec: inverse mapping} holds for any carrier frequency. In this section, we provide illustrations for some specific carrier frequencies  currently used in the 5G standards~\cite{9566671}.
\begin{figure}[htbp]
	\centering
	\includegraphics[width=0.47\textwidth]{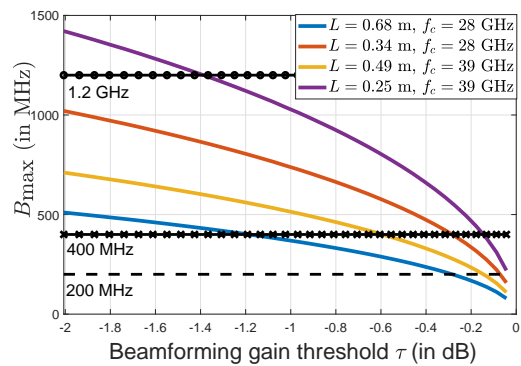}
		
	\caption{$B_{\mathsf{max}} $ as a function of the beamforming gain threshold for $\{{L}, f_\mathsf{c}\} \equiv $  (0.68 m, 28 GHz), (0.34 m, 28 GHz), (0.49 m, 39 GHz), (0.25 m, 39 GHz). Tradeoff between gain threshold and maximum usable bandwidth is illustrated.}
	\label{fig:bw vs array gain threshold}
\end{figure}

 \begin{figure}[htbp]
 	
 	\centering
 	\includegraphics[width=0.47\textwidth]{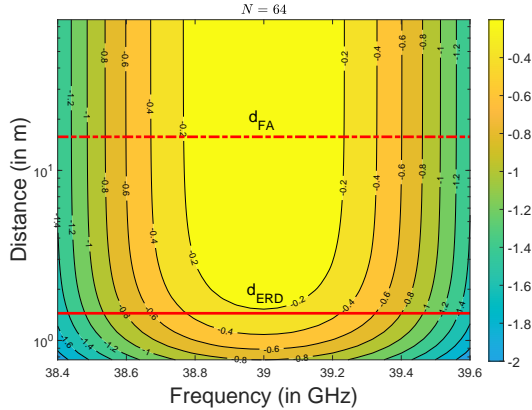}
 	\vspace{-11pt}
 	\caption{Contour plot of beamforming gain (in dB) as a function of distance (in m) and frequency (in GHz) for $f_\mathsf{c}=39$ GHz and $N=64$.  The $\mathsf{BAND}$ attains a global minima at $f_\mathsf{c}$.  }
 	\label{fig:2d plot contour r vs f}
 	
 \end{figure}
In Fig.~\ref{fig:bw vs array gain threshold},  
$B_{\mathsf{max}}$ is plotted as a function of the beamforming gain threshold for $\theta_{\mathsf{worst}} =60^{\circ} $.
For n261 band in 5G NR\cite{9566671},  $f_\mathsf{c}=28$ GHz.
For n260 band in 5G NR\cite{9566671}, $f_\mathsf{c}=39$ GHz.
For each band, we plot $B_{\mathsf{max}}$ for $N=\{64, 128\}$.
 We also mark the  bandwidths of 200 MHz, 400 MHz, and 1.2 GHz.

Fig.~\ref{fig:bw vs array gain threshold} offers two important insights. For a given $\{{L}, f_\mathsf{c}\}$ pair, we can determine the maximum possible beamforming gain for the 5G NR bandwidths.  Conversely, we can also determine $B_{\mathsf{max}} $ for a fixed value of beamforming gain. As expected from \eqref{eqn: upper bound abp},  for a fixed $f_\mathsf{c}$, $B_{\mathsf{max}} $  doubles as ${L}$ gets halved to maintain the same beamforming gain.

In Fig.~\ref{fig:2d plot contour r vs f}, we show the  contour plot in Fig.~\ref{fig:2d plot gamma1 gamma2 product} with a change of variables from $(\gamma_1, \gamma_2)$  to the $({r}, {f})$ space using \eqref{eqn: bar f} and \eqref{eqn: bar r} for  $\bar{d}=0.5$, $\theta=60^{\circ}$, $f_\mathsf{c}=39$ GHz and for $N=64$. 
Each contour denotes the $\mathsf{BAND}$ for a given gain threshold.
For operating distances greater than the $\mathsf{BAND}$, the beamforming gain will always remain above the threshold. The plot shows that the distance increases for frequencies away from the center frequency. This distance diverges beyond a certain value of frequency $|f|$, which illustrates the concept of the maximum usable bandwidth
 derived in \eqref{eqn: upper bound abp}. 
In Fig.~\ref{fig:2d plot contour r vs f}, we also plot $d_{\mathsf{ERD}}$ and $d_{\mathsf{FA}}$ which are same for all frequencies in the band.  
The distance $d_{\mathsf{FA}}$ is inaccurate because it does not incorporate the angle dependence and frequency-selectivity.
The distance $d_{\mathsf{ERD}}$ is consistent with the $\mathsf{BAND}$ value for $\tau=-0.2$ dB only at the center frequency 39 GHz because $d_{\mathsf{ERD}}$ is derived using a narrowband model.
The distance $d_{\mathsf{ERD}}$ underestimates the near-field distance 
 for frequencies away from $f_\mathsf{c}$.
% Hence, as $\bar{f}$ shifts  away from  $f_\mathsf{c}$, the $\mathsf{BAND}$ also increases.

%\section{Simulation results}
\section{Conclusion}
In this letter, we proposed a new definition of the beamforming gain metric which incorporates both wideband and near-field propagation effects. For a MISO system with a ULA at the transmitter, we provided  a simple closed-form expression for the beamforming gain in terms of standard Fresnel functions with  two parameters that model the near-field and wideband effects. A key observation is that the beamforming gain depends only on the normalized frequency and normalized distances, which enables the validity of the insights for any carrier frequency. The proposed upper bound on the aperture-bandwidth product is beneficial for characterizing the performance of the existing frequency-flat beamforming when scaling up in carrier frequency, bandwidth, and array aperture. 
We showed that the $\mathsf{BAND}$ corresponding to a particular threshold attains minima at  $f_\mathsf{c}$ and increases for frequencies away from $f_\mathsf{c}$. The model and analysis presented in this work is especially relevant for short distance transmission where the near-field effect is more relevant, with potentially small impact from angular spread. We encourage future studies to address the impact of angular spread on beamforming gain.
 Other future directions of work include extending the $\mathsf{BAND}$ definition to planar arrays, the MIMO LOS channel, and incorporating the mutual coupling effect for dense arrays.

\bibliographystyle{IEEEtran}
\bibliography{references_Nitish.bib}

\end{document}